\documentclass[article,a4paper,oneside,11pt]{memoir}

\counterwithout{section}{chapter}

\usepackage[T1]{fontenc}
\usepackage[utf8]{inputenc}
\usepackage{amsmath,amssymb,amsthm}
\usepackage{enumerate}
\usepackage{color}
\usepackage{mathtools}
\usepackage[ruled]{algorithm2e}
\usepackage{thm-restate}
\usepackage{tikz}

\usepackage{geometry}
\renewcommand{\epsilon}{\varepsilon}

\DeclarePairedDelimiter{\ceil}{\lceil}{\rceil}

\newtheorem{theorem}{Theorem}[section]
\newtheorem{lemma}[theorem]{Lemma}

\newtheorem{corollary}[theorem]{Corollary}

\newtheorem{claim}{Claim}

\title{A general lower bound for collaborative tree exploration}

\newcommand*\samethanks[1][\value{footnote}]{\footnotemark[#1]}

\author{Yann Disser%
\thanks{Institute of Mathematics, TU Darmstadt, Germany. disser@mathematik.tu-darmstadt.de}%
\thanks{Supported by DFG SPP 1736.}
\and
Frank Mousset%
\thanks{Department of Computer Science, ETH Zurich, Switzerland. 
\newline {\{moussetf|anoever|nskoric|steger\}@inf.ethz.ch.} }%
\thanks{Supported by grant no.~6910960 of the Fonds National de la Recherche, Luxembourg.}
\and
Andreas Noever%
\samethanks[3]%
\thanks{Supported by grant no.~200021 143338  of the Swiss National Science Foundation.}
\and
Nemanja \v{S}kori\'c\samethanks[3]
\and
Angelika Steger\samethanks[3]}

\date{\today}

\begin{document}
	\maketitle

\begin{abstract}
  We consider collaborative graph exploration with a set of~$k$ agents.
  All agents start at a common vertex of an initially unknown graph and need to
  collectively visit all other vertices. We assume agents are deterministic,
  vertices are distinguishable, moves are simultaneous, and we allow agents to
  communicate globally. For this setting, we give the first non-trivial lower
  bounds that bridge the gap between small ($k \leq \sqrt{n}$) and large ($k
  \geq n$) teams of agents. Remarkably, our bounds tightly connect to existing
  results in both domains. 

  First, we significantly extend a lower bound of~$\Omega(\log k / \log
  \log k)$ by Dynia et al.~on the competitive ratio of a collaborative tree
  exploration strategy to the range~$k \leq n \log^c n$ for any~$c \in
  \mathbb{N}$.
	Second, we provide a tight lower bound on the number of agents needed for any 
	competitive exploration algorithm. In particular, we show that
	any collaborative tree exploration algorithm with~$k=Dn^{1+o(1)}$
	agents has a competitive ratio of~$\omega(1)$, while Dereniowski et
	al.~gave an algorithm with $k=Dn^{1+\varepsilon}$ agents and
	competitive ratio~$\mathcal{O}(1)$, for any~$\varepsilon > 0$ and
	with~$D$ denoting the diameter of the graph.
  Lastly, we show that, for any exploration algorithm using $k=n$ agents,
  there exist trees of arbitrarily large height $D$ that
  require~$\Omega(D^2)$ rounds, and we provide a simple
  algorithm that matches this bound for all trees.
\end{abstract}

\section{Introduction}

Graph exploration captures the problem of navigating an unknown terrain with a
single or multiple autonomous robots. In the abstract setting, we take the
perspective of an agent that is located at
some vertex of an initially unknown graph, can locally distinguish edges at its
current location, and can choose an edge to traverse in its next move. Various
scenarios for graph exploration have been studied in the past, for different
graph classes and different capabilities of the agent(s). A fundamental goal of
exploration is to systematically visit all vertices/edges of the underlying
graph.
For settings where exploration is possible, we typically ask for efficient exploration algorithms, e.g., in terms of the number of edge traversals.

In this paper, we consider \emph{collaborative} exploration, where a set of~$k$
agents are initially located at some vertex of an unknown \emph{undirected}
graph. We assume agents to move deterministically, allow them to freely
communicate at all times, and to have unlimited computational power and memory
at their disposal. In every round each agent may traverse any edge incident to
its current location, where the edges incident to a vertex are revealed when
that vertex is visited for the first time. The goal is to visit all vertices
while minimizing the number of rounds. More precisely, we are interested in the
competitive ratio of an exploration strategy, i.e., the worst case ratio
between the total number of rounds it needs and the minimum total number of
rounds needed to visit all vertices of the same graph, assuming it is known
beforehand.
We prove new lower bounds for the best-possible competitive ratio of any
collaborative exploration algorithm. Our bounds hold even for the much simpler
setting of \emph{tree} exploration. Note that since our results concern trees,
it makes no difference whether nodes can be distinguished, and whether the
agents need to visit all edges or not.

Let~$\mathcal T_{n,D}$ denote set of all rooted trees with $n$ vertices and
height $D$. Each such tree corresponds to an instance of the tree exploration
problem in which all~$k$ agents start at the root of the tree. Clearly, any
offline exploration algorithm needs $\Omega(n/k+D)$ rounds to explore a tree in
$\mathcal T_{n,D}$ using $k$ agents. This is shown to be tight by the following
offline exploration algorithm that explores the tree in $\Theta(n/k+D)$ rounds:
start with the tree $T$, double its edges, find an Eulerian tour $C$ (of length
$2n-2$), distribute the agents evenly on $C$ (this takes at most $D$ rounds),
and explore $T$ by letting each agent walk along $C$ for $\mathcal O(n/k)$
rounds.

In the online setting, we can
explore a tree in $\mathcal T_{n,D}$ with a single agent using a depth-first
traversal in time $\mathcal{O}(n)$ and thus we trivially have a competitive
ratio of~$\mathcal{O}(1)$ when $k$ is constant. On the other hand, with~$k\geq
\Delta^{D}$ agents, where~$\Delta$ is the maximum degree of the tree, we can
simply perform a breadth-first traversal, which takes~$\mathcal{O}(D)$ steps
and thus also has competitive ratio~$\mathcal{O}(1)$. Observe that in the first
case~$n/k$ dominates the lower bound on the offline optimum, while in the
second case~$D$ is dominating. We are interested in the best-possible
competitive ratios between these two extreme cases.

Surprisingly, Dereniowski et
al.~\cite{DereniowskiDisserKosowskiPajakUznanski/15} showed that already a
polynomial number $k=Dn^{1+\varepsilon}$ of agents allows for a BFS-like
algorithm that achieves a constant competitive ratio. For smaller teams of
agents, Fraigniaud et al.~\cite{FraigniaudGasieniecKowalskiPelc/06,
HigashikawaKatohLangermanTanigawa/12} gave a collaborative algorithm with
competitive ratio~$\mathcal{O}(k/\log k)$. This is only slightly better than
the trivial upper bound of~$\mathcal{O}(k)$ that we get by performing a depth
first traversal with a single agent. Ortolf and
Schindelhauer~\cite{OrtolfSchindelhauer/14} improved this competitive ratio
to~$k^{o(1)}$ for~$k=\smash{2^{\omega(\sqrt{\log D \log \log D})}}$ and~$n =
\smash{2^{\mathcal{O}(2^{\sqrt{\log D}})}}$. The only non-trivial lower bound
for collaborative tree exploration was given by Dynia et
al.~\cite{DyniaLopuszanskiSchindelhauer/07}. They showed that any deterministic
exploration algorithm for~$k<\sqrt{n}$ agents has competitive
ratio~$\Omega(\log k /\log \log k)$.

\subsection*{Our Results}

\begin{figure}
\centering
\begin{tikzpicture}
[
  upperb/.style={black!30!green, very thick},
  breakp/.style={circle, fill, minimum size=4, inner sep = 0}
]
\def\hyt{6}
\def\wid{12}
\def\one{1.02}

\coordinate (b1) at (0,\one);
\coordinate (bu2) at (2.8,1.95);
\coordinate (bu3) at (9,3.25);
\coordinate (bl1) at (3.4,1.43);
\coordinate (bl2) at (7.5,1.75);
\coordinate (bl3) at (7.5,2*\one-1);
\coordinate (bl4) at (9,\one);
\coordinate (bl5) at (\wid-.5,\one);

\draw [dotted] (b1) -- ++(0,-\one);
\node [] at (0,-.4) {$\mathcal{O}(1)$};
\draw [dotted] (bu2) -- ++(0,-1.95);
\node [] at (2.8,-.4) {$D^{\varepsilon}$};
\draw [dotted] (bu3) -- ++(0,-3.25);
\node [] at (3.4,-.4) {$\sqrt{n}$};
\draw [dotted] (bl1) -- ++(0,-1.5);
\node [] at (7.5,-.4) {$n \log^c n$};
\draw [dotted] (bl2) -- ++(0,-1.75);
\node [] at (9,-.4) {$D n^{1+\varepsilon}$};
\draw [dotted] (bl5) -- ++(0,-\one);
\node [] at (\wid-.5,-.4) {$\Delta^D$};

\draw [upperb] (0,\one) -- node [left] {$\mathcal{O}(1)$} (b1);
\draw [black!30!green, dashed] (0,\one) -- node [sloped, very near end,above] {$\mathcal{O}(k)$} (\wid,\hyt-.5);
\draw [upperb] (b1) -- node [sloped,anchor=south,auto=false] {$k/\log k$ \cite{FraigniaudGasieniecKowalskiPelc/06,
HigashikawaKatohLangermanTanigawa/12}} (bu2);
\node [rotate=10, black!30!green, fill=white] at (6, 3.2) {$k^{o(1)}$ \cite{OrtolfSchindelhauer/14}};
\draw [upperb] plot [smooth, tension=1] coordinates {(bu2) (6.1,2.85) (bu3)};
\draw [upperb] (bl4) -- node [sloped,anchor=south,auto=false] {$\mathcal{O}(1)$ \cite{DereniowskiDisserKosowskiPajakUznanski/15}} (bl5);
\draw [upperb] (bl5) -- (\wid,\one);

\draw [black!30!red, dashed, thick] (0,1) -- node [very near end, below] {$\Omega(1)$} (\wid,1);
\draw [black!30!red, very thick] plot [smooth, tension=.55] coordinates {(b1) (2, 1.27) (bl1)};
\draw [black!30!red, line width=2.5pt] plot [smooth, tension=.55] coordinates {(bl1) (6, 1.67) (bl2)};
\node [rotate=5.5, black!30!red, fill=white] at (3.75, 1.05) {$\Omega(\log k / \log \log k)$ \cite{DyniaLopuszanskiSchindelhauer/07}};
\draw [black!30!red, line width=2.5pt] (bl3) -- node[below] {$\omega(1)$} (9,2*\one - 1);

\node [breakp] at (b1) {};
\node [breakp] at (bu2) {};
\node [breakp] at (bu3) {};
\node [breakp] at (bl1) {};
\node [breakp] at (bl2) {};
\node [breakp] at (bl3) {};
\node [breakp] at (bl4) {};
\node [breakp] at (bl5) {};

\draw [black, very thick, ->]  (0,0) -> node [sloped, very near end, anchor=south, auto=false] {comp. ratio} (0,\hyt);
\draw [black, very thick, ->]  (0,0) -> (\wid,0);
\node at (\wid+.2,0) {$k$};


\end{tikzpicture}
\caption{State of the art in collaborative tree exploration.
Upper bounds are green and lower bounds are red. Thick lines show our results. \label{fig:state-of-the-art}}
\end{figure}
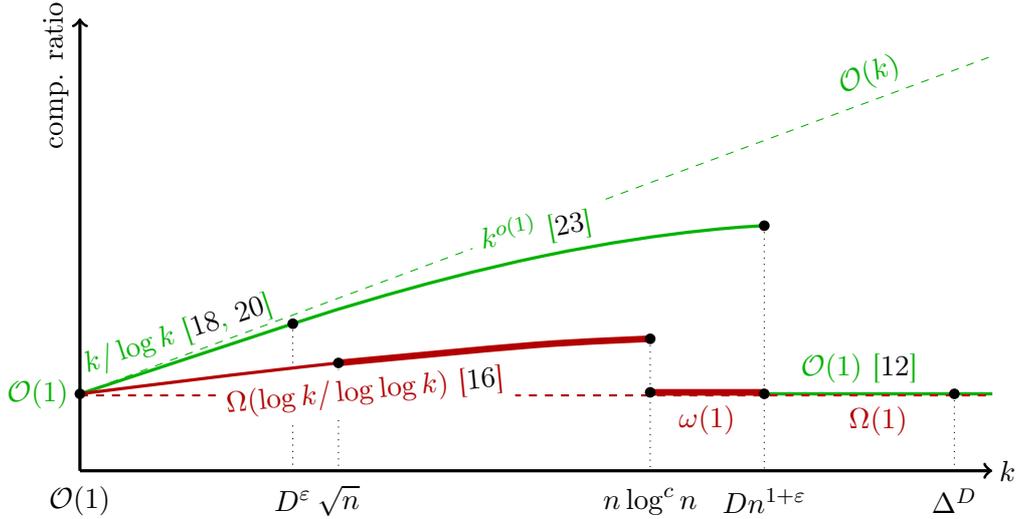

We give the first non-trivial lower bounds on the competitive ratio for
collaborative tree exploration in the domain~$k \geq \sqrt{n}$ (cf.~Figure~\ref{fig:state-of-the-art}).
More precisely, we show that for every constant $c\in \mathbb N$,
any given deterministic exploration strategy with $k \leq n\log^c n$ agents
has competitive ratio $\Omega(\log k /\log \log k)$ on the set of all trees
on $n$ vertices. Note that this extends the range of the bound by Dynia et
al.~\cite{DyniaLopuszanskiSchindelhauer/07} for~$k < \sqrt{n}$ significantly.

Secondly, we show that for every constant $\varepsilon>0$,
there is a constant $D=D(\epsilon)$ such that for any exploration algorithm
with~$k\leq Dn^{1+\varepsilon}$ agents, there exists a tree in $\mathcal T_{n,D}$
on which the algorithm needs at least $D/(5\varepsilon)$ rounds. 
This (almost) tightly matches the algorithm of Dereniowski et
al.~\cite{DereniowskiDisserKosowskiPajakUznanski/15}, which
can explore any tree in at most $(1+o(1))D/\varepsilon$ 
rounds using $k=Dn^{1+\epsilon}$ agents. 
Our result implies that any exploration algorithm with $k=D^{1+o(1)}$ agents has
competitive ratio~$\omega(1)$.
More precisely, we get that for any function $0\leq f(n)\leq o(1)$,
there is a function $D=D(n)$ such that every exploration algorithm 
with~$k = Dn^{1+f(n)}$ agents has competitive ratio~$\omega(1)$ on the trees 
in~$\mathcal T_{n,D}$.
In contrast, the algorithm of Dereniowski et al.~shows
that $k=Dn^{1+\varepsilon}$ agents are sufficient
to get a competitive ratio~$\mathcal{O}(1)$ on such trees.

Finally, for every exploration algorithm with~$k=n$, we construct a tree of height~$D = \omega(1)$ where the algorithm needs~$\mathcal{O}(D^2)$ rounds.
We give a simple algorithm that achieves this bound in general.

\subsection*{Further Related Work}


Many variants of graph exploration with a \emph{single agent} have been studied
in the past. 
Any (strongly) connected graph with \emph{distinguishable} vertices can easily be explored in polynomial time by systematically building a map of the graph.
Regarding the exploration of \emph{undirected} graphs with \emph{indistinguishable} vertices, Aleliunas et al.~\cite{AleliunasKarpLiptonLovaszRackoff/79} showed that a random walk explores any graph in~$\mathcal{O}(n^3 \Delta^2 \log n)$ steps, with high probability.
In order to turn this into a terminating exploration algorithm the agent needs~$\Omega(\log n)$ bits of memory.
Fraigniaud et al.~\cite{FraigniaudIlcinkasPeerPelcPeleg/05} showed that every deterministic algorithm needs~$\Omega(\log n)$ bits of memory, and Reingold~\cite{Reingold/08} gave a matching upper bound.
Disser et al.~\cite{DisserHackfeldKlimm/16} showed that alternatively~$\Theta(\log \log n)$ pebbles and bits of memory are necessary and sufficient for exploration, where a pebble is a device that can be dropped to make a vertex distinguishable and that can be
picked up and reused later.
Diks et al.~\cite{DiksFraigniaudKranakisPelc/04} showed that trees can be explored with $\mathcal{O}(\log \Delta)$ memory, and that $\Omega(\log n)$ memory is required if the agent needs to eventually terminate at the start vertex. 
Ambühl~\cite{AmbuhlGasieniecPelcRadzikZhang/11} gave a matching upper bound for the latter result.

For the case of \emph{directed} graphs with \emph{distinguishable} vertices, Albers and
Henzinger~\cite{AlbersHenzinger/00} gave an exploration algorithm with subexponential
running time $d^{\mathcal{O}(\log d)}m$ that learns a map of the graph.
Here~$m$ denotes the number of edges and~$d$ is the deficiency of the graph, i.e., the number of edges missing to
make the graph Eulerian. This results narrows the gap between a quadratic lower bound
and an exponential upper bound introduced by Deng and
Papadimitriou~\cite{DengPapadimitriou/99}. 

An even more challenging setting (for the agent) is the exploration of \emph{directed},
strongly connected graphs with \emph{indistinguishable} vertices. In general the agent
needs exponential time to explore a graph in this setting. On the other hand,
Bender and Slonim~\cite{BenderSlonim/94} showed that two agents can explore any
directed graph in polynomial time, using a randomized strategy. Bender et
al.~\cite{BenderFernandezRonSahaiVadhan/02} showed that to accomplish this with
a single agent we need~$\Theta(\log \log n)$ pebbles, i.e., ``a friend is worth
$\mathcal{O}(\log \log n)$ pebbles''. Remarkably, Bender et
al.~\cite{BenderFernandezRonSahaiVadhan/02} also showed that if the number of
vertices is known beforehand, a deterministic agent with a single pebble can
explore any directed graph in polynomial time~$\mathcal{O}(n^8 \Delta^2)$.

The lower bounds for collaborative tree exploration discussed above carry over to the \emph{collaborative exploration} of general undirected graphs with distinguishable vertices.
Also, the algorithm of Dereniowski et al.~\cite{DereniowskiDisserKosowskiPajakUznanski/15} for $k = Dn^{1+\varepsilon}$ works on general graphs.
Additionally, Ortolf and Schindelhauer~\cite{OrtolfSchindelhauer/12} gave a lower bound on the best-possible competitive ratio for randomized algorithms of~$\Omega(\sqrt{\log k}/\log \log k)$ for $k = \sqrt{n}$.
Collaborative exploration by multiple random walks without communication has been considered by Alon et al.~\cite{AlonAvinKouckyKozmaLotkerTuttle/11}, Elsässer and Sauerwald~\cite{ElsasserSauerwald/11}, and Ortolf and Schindelhauer~\cite{OrtolfSchindelhauer/15}.

Graph exploration has been studied in many other settings.
Examples include tethered exploration or exploration with limited fuel~\cite{AwerbuchBetkeRivestSingh/99, DuncanKobourovKumar/06}, exploration of mazes~\cite{BlumKozen/78, Hoffmann/81}, and exploration of polygonal environments~\cite{ChalopinDasDisserMihalakWidmayer/11, ChalopinDasDisserMihalakWidmayer/15}.

\section{Results}

Our first result extends the lower bound for $k<\sqrt{n}$ agents of Dynia et
al.~\cite{DyniaLopuszanskiSchindelhauer/07} to the much larger range $k\leq
n\log^{\mathcal O(1)}n$. We prove the following theorem:

\begin{restatable}{theorem}{thmLoglogCompRatio}
	\label{thm1}
  Let $c$ be any positive integer constant.
  Then for every $n$ and every $1\leq k \leq n\log^c n$ there is some $D =
  D(n,k,c)$ such that the following holds: for any given deterministic
  exploration strategy with $k$ agents, there exists a tree $T$ on $n$ vertices
  and with height $D$ on which the strategy needs
  \[ \Omega\Big(\,\frac{\log k}{\log \log k}\cdot(n/k+D)\,\Big) \]
  rounds.
\end{restatable}

As mentioned above, there is an offline algorithm that explores any graph with
$n$ vertices and height $D$ in time $\Theta(n/k+D)$. From this, we obtain the
following corollary to Theorem~\ref{thm1}:

\begin{corollary}
  Let $c$ be any positive integer constant. Then any deterministic exploration
  strategy using $k \leq n\log^c n$ has a competitive ratio of
  \[ \Omega\Big(\,\frac{\log k}{\log \log k}\,\Big). \]
\end{corollary}

Our second main result shows that the algorithm of Dereniowski et
al.~\cite{DereniowskiDisserKosowskiPajakUznanski/15} that explores a graph
with $k = Dn^{1+\epsilon}$ agents in time $(1+o(1))D/\epsilon$ is almost
optimal: using $k\leq Dn^{1+\epsilon}$ agents it is generally impossible to
explore the graph in fewer than $D/(5\epsilon)$ rounds.

\begin{restatable}{theorem}{thmIterations}
	\label{thm2}
  Given any constant $\epsilon>0$ there is an integer $D=D(\epsilon)$ such that
  for sufficiently large $n$ and for every deterministic exploration strategy
  using $k \leq D\cdot n^{1+\epsilon}$ agents, there exists a tree on $n$
  vertices and with height $D$ on which the strategy needs at least
  $D/(5\epsilon)$ rounds.
\end{restatable}

In the range where $k \geq n$, the offline optimum is determined by the height~$D$ 
of the tree. Therefore, the result of Dereniowski et al.\ mentioned above
implies that the competitive ratio is constant when $k = D\cdot
n^{1+\Omega(1)}$. Theorem~\ref{thm2} shows in particular that in some
sense this is tight:

\begin{corollary}
  For any function $0 \leq f(n)\leq o(1)$, there is a function
  $D=D(n)$ such that the competitive ratio of any deterministic
  exploration strategy using $k = D\cdot n^{1+f(n)}$ agents 
	is~$\omega(1)$ on the the set~$\mathcal T_{n,D}$ of all rooted trees
  with $n$ vertices and height $D$. 
\end{corollary}

Finally, it is possible for $k=n$ agents to explore any tree on $n$
vertices and of height $D$ in $D^2$ rounds using a breadth-first
exploration strategy. More precisely, we can split the $D^2$ rounds in
$D$ phases of length $D$, and in each phase $1\leq i\leq D$ do the
following. Let $A_i$ be the set of unvisited leaves of the tree that is
revealed at the start of phase $i$. Then we send one agent to each
vertex in $A_i$ along a shortest path. This is clearly doable in $D$
rounds, and after $D$ such phases, the tree is completely explored. We
show that the running time of $D^2$ is optimal up to a constant factor:

\begin{restatable}{theorem}{thmSquaredRounds}
	\label{thm3}
  For every $n$ and every deterministic exploration strategy using $k=n$
  agents, there exists a tree $T$ on $n$ vertices and with height $D=\omega(1)$
  such that the strategy needs at least $D^2/3$ rounds to explore $T$.
\end{restatable}

In all the results above, we have considered the worst-case performance
of an exploration strategy on any tree. However, by looking at the proofs
of Theorems~\ref{thm1} and \ref{thm2}, one can see that
the heights of our lower bound constructions are typically quite small.
We believe it is also natural to ask about the competitive ratio on
the set of trees of height at least $D$, for a given $D$. We show that at
least for subpolynomial heights, the competitive ratio with $k=\Theta(n)$
agents is unbounded:

\begin{restatable}{theorem}{thmAnySublinearDiameter}
  \label{thm:any}
  For any function $D\leq n^{o(1)}$ and any exploration strategy using $k=\Theta(n)$
  agents, the competitive ratio on the set of all trees of size~$n$ and height 
	at least~$D$ is~$\omega(1)$.
\end{restatable}

We stress that Theorem~\ref{thm:any} differs from the other results in that
it applies to \emph{any} height $D\leq n^{o(1)}$, while in the other results
merely state that there \emph{exists} some height with the desired property.

%
%

\section{Tree exploration games}\label{sec:games}

In order to prove a lower bound on the competitive ratio, we consider a tree exploration game defined as follows.
By a \emph{tree exploration game with $k$ agents} we mean a game with two
players, the \emph{explorer} (the online algorithm) and the \emph{revealer} (the adversary), 
played according to the
following rules. The game proceeds in rounds which we index by the variable $t$
(`time'), the first round being $t=0$. The state of the game at time $t$ is
described by a triple $(T_t, A_t, \phi_t)$, where $T_t$ is a rooted tree (the
tree revealed at the beginning of round $t$), $A_t$ is a subset of the vertices
of $T_t$ (the subset of \emph{visited} vertices by  round $t$),
and $\phi_t\colon \{1,\dotsc,k\} \to A_t$ is an assignment of the agents to the vertices
(where $\phi_t(i)$ is the location of the $i$-th agent at time $t$). In round
$t=0$ the revealer decides on the initial tree $T_0$. The state at time $0$ is
then given by $(T_0, A_0, \phi_0)$ where $A_0= \{\text{root}(T_0)\}$ and
$\phi_0(x)= \text{root}(T_0)$ for all $1\leq x\leq k$ -- that is to say, all agents
are initially at the root of $T_0$. In every round $t>0$, each player can make
a move. First, the explorer creates a new assignment $\phi_t$ by moving 
each agent $i$ to a neighbor of $\phi_{t-1}(i)$ in $T_{t-1}$ or by keeping the location of the agent same, i.e., $\phi_t(i) = \phi_{t-1}(i)$.
Then the revealer decides on the new tree $T_t$, where $T_t$ must be obtained
from $T_{t-1}$ by attaching (possibly empty) trees at some vertices $v\in
V(T_{t-1})\setminus A_{t-1}$, where $V(T_{i-1})$ is the set of vertices of $T_{i-1}$. We then let $A_t = A_{t-1} \cup N_t$ where  $N_t
= \{ \phi_t(i) : 1\leq i\leq k\}$ is the set of the new agent locations. The
game ends in round $t^*$ if all vertices of $T_{t^*}$ are visited at the
beginning of round $t^*$, i.e., if $A_{t^*} = V(T_{t^*})$.

This type of game naturally lends itself to proving lower bounds for the
time in which $k$ agents can explore an unknown tree. Specifically, consider
any deterministic strategy for exploring an unknown tree $T$ with $k$ agents.
Such a strategy can be interpreted as a strategy for the explorer in the tree
exploration game with $k$ agents. If the revealer can play so that the game
lasts for at least $t^*$ rounds, then this means that the proposed exploration
strategy needs $t^*$ rounds to explore the tree $T_{t^*}$. We will use this
observation to prove lower bounds for the online graph exploration in the
following section.

As a side remark, here it is crucial that the strategy is
deterministic: if the strategy were allowed to make random choices, then the
tree $T_{t^*}$ would turn out to be a random variable that might be highly
correlated with the random choices made by the explorer, and it could not serve
as an instance on which the strategy performs badly.

\section{Lower bound construction}\label{sec:main}
We now give our lower bound construction that establishes the following technical lemma.

\begin{lemma}\label{lemma:main}
  Let $n,L,m$ be positive integers such that $n\geq L\cdot 16^m$.
  Then for any deterministic exploration strategy using
  \[ k\leq \frac{n^{1+1/m}}{6L(m+1)^2(2L)^{1/m}} \]
  agents, there exists a tree $T$ on $n$ vertices and of height $Lm$
  such that the strategy needs at least $L\binom{m}{2}$ rounds to explore $T$.
\end{lemma}

\begin{proof}Assume that integers $n$, $L$ and $m$ as above are given.
Let $k$ be any integer such that $1\leq k \leq n^{1+1/m}/(6L(m+1)^2(2L)^{1/m})$. 
To prove the lemma, we will describe a strategy for the revealer in the tree
exploration game with $k$ agents such that
\begin{itemize}
  \item the game does not end before round $t^* := L\cdot \binom{m}{2}$, and
  \item the tree $T_{t^*}$ has height $Lm$ and at most $n$ vertices,
\end{itemize}
where the notation is as in Section~\ref{sec:games}. Note that this is enough
to prove the lemma.

Before explaining the strategy, we fix some notation. Let
\[ \alpha := (2L/n)^{1/m} \quad \text{ and } \quad t_i := L\cdot \binom{i+1}{2} \]
for $0\leq i < m$.
For each $t\geq 0$ we can consider the equivalence relation $\sim_t$ on
$V(T_{t})$ where $u\sim_t v$ if there exists a path between $u$ and $v$ in
$T_{t}$ that avoids the root of $T_{t}$ (i.e., if they have a common
ancestor that is not the root). 
Since $T_0\subseteq T_1\subseteq T_2 \subseteq \dotsc$ are trees with the same root, we will just write 
$u\sim v$ instead of $u\sim_t v$ without causing confusion.
Then we define 
  $a_t(v) := |\{x \mid \phi_t(x) \sim v\}|$. 
In other words, $a_t(v)$ counts the total number of agents
that could reach vertex $v$ without passing through the root (under the assignment $\phi_t$).
We give the strategy for the revealer in Algorithm~\ref{alg:revealer}.

\vskip\baselineskip

\begin{algorithm}[H]
  \caption{The strategy for the revealer. \label{alg:revealer}}
  \SetAlgoLined
  \Begin{
    let $T_0$ be a `star' consisting of $\ceil{n/(2L)}$ paths 
      of length $L$ from the root\;
    \ForEach{round $t = 1,2,3, \dotsc$}{
      let the explorer choose $\phi_t$\;
      \eIf{$t = t_i$ for some $1\leq i< m$}{
        let $K_i$ be a maximal set of vertices in $V(T_{t_i-1})\setminus A_{t_i-1}$ s.t.
        \begin{enumerate}[(i)]
          \item every vertex in $K_i$ has distance $L\cdot i$ to the root in $T_{t_i-1}$
          \item there are no two distinct vertices $u,v\in K_i$ with $u\sim v$.
        \end{enumerate}
      let $S_i\subseteq K_i$ be the $\ceil{\alpha |K_i|}$ vertices
      $v\in K_i$ with least  $a_{t_i}(v)$\;
      define $T_{t_i}$ by attaching at each
      $v\in S_i$ a path of length $L-1$ with a
      star with $L\cdot (i+1) \cdot a_{t_i}(v)$ leaves at the end\;
      }{let $T_t = T_{t-1}$\;}
    }
  }
\end{algorithm}

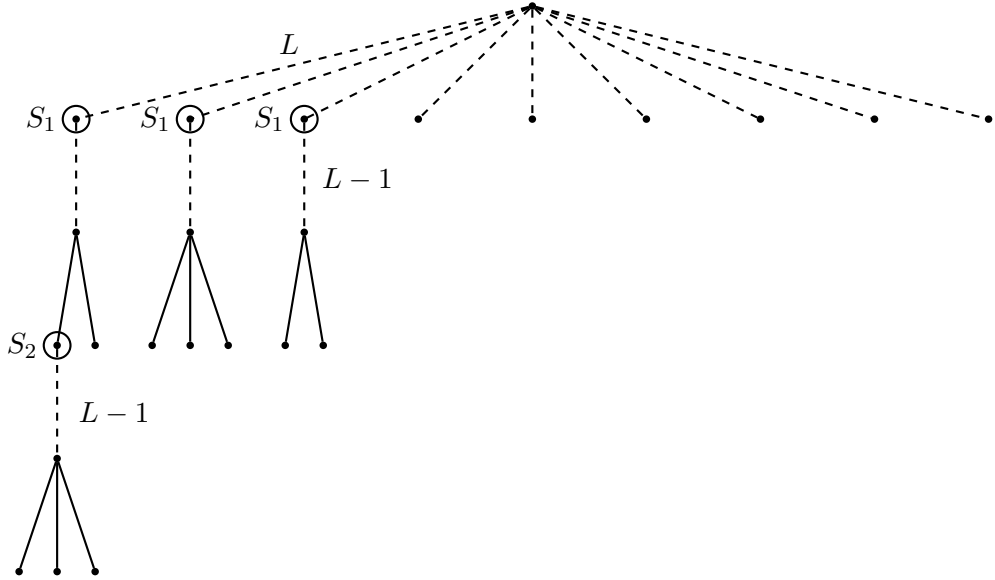
\begin{figure}[h]
  \label{fig}
  \centering
    \begin{tikzpicture}[level distance = 1.5cm]
      \tikzset{every path/.style={thick}}
      \tikzset{every node/.style={solid,draw,fill=black,circle,minimum size=2pt,inner sep=0pt}}
      \tikzset{s/.style={circle,fill=none,minimum size=10pt}}
      \path (-4,-4) rectangle (4,5);
      \node at (0,4) {}
         child [dashed] {[sibling distance=0.5cm] node (a) {}
                 child {node {} child [solid] {node (c) {} child [dashed] {node {} child [solid] {node {}}
                                                                  child [solid] {node {}}
                                                                  child [solid] {node {}}}}
                                child [solid] {node {}}}
               }
         child [dashed] {[sibling distance=0.5cm] node (b) {}
                 child {node {} child [solid] {node {}}
                                child [solid] {node {}}
                                child [solid] {node {}}}
               }
         child [dashed] {[sibling distance=0.5cm] node (d) {}
                 child {node {} child [solid] {node {}}
                                child [solid] {node {}}}
               }
         child [dashed] {node {}}
         child [dashed] {node {}}
         child [dashed] {node {}}
         child [dashed] {node {}}
         child [dashed] {node {}}
         child [dashed] {node {}};
      \node[s,label=180:$S_1$] at (a) {};
      \node[s,label=180:$S_1$] at (d) {};
      \node[s,label=180:$S_1$] at (b) {};
      \node[s,label=180:$S_2$] at (c) {};
      \tikzset{every node/.style={draw=none,fill=none}};
      \node at (-3.2,3.5) {$L$};
      \node at (-2.3,1.7) {$L-1$};
      \node at (-5.5,-1.4) {$L-1$};
    \end{tikzpicture}
  \caption{A sketch of the tree generated by the revealing strategy, for artificial
  values $\alpha=1/3$ and $\ceil{n/(2L)}=9$ (degrees in the actual construction are much larger). 
	The actual shape depends on the distribution of the agents at times $t_1, t_2$. 
	Dashed lines represent paths of the specified length.}
\end{figure}

\vskip\baselineskip

For a better intuition, we refer the reader to Figure~\ref{fig}, which shows
what the tree constructed by this strategy might look like. We establish three
claims which are used to show that the algorithm indeed runs for at least $t^*
= t_{m-1}$
rounds and that the tree constructed in this way has the right properties.

\begin{claim}\label{cl:height}
  For every $0\leq i <m$ the following holds.
  The height of $T_{t_i}$ is at most $L\cdot (i+1)$. Moreover, if
  $S_1,\dotsc,S_i$ are all non-empty, then the height of $T_{t_i}$ is exactly
  $L\cdot (i+1)$.
\end{claim}

\begin{proof}
  The tree $T_{t_i}$ differs from $T_{t_{i-1}}$ if and only if $S_i$ is non-empty, and in this case
  it is obtained by attaching trees of height $L$ at some vertices with distance $L\cdot i$
  to the root in $T_{t_{i-1}}$. Since $T_{t_0} = T_0$ has height $L$, this implies the claim
  by induction.
\end{proof}

\begin{claim}\label{cl:cl2}
  For all $1\leq i < m$ and every $v\in S_i$, there exists at least one
  descendant of $v$ at depth $L\cdot (i+1)$ in $T_{t_{i+1}-1}$ that does not belong to $A_{t_{i+1}-1}$.
  In particular, for all $1\leq i < m$ we have $|K_{i+1}| = |S_i|$. 
\end{claim}

\begin{proof}
  Each vertex at depth $L(i+1)$ is a descendant of some vertex $v\in S_i$. Moreover, we have $u\nsim v$
for any two distinct $u,v\in S_i$. Thus, the second claim follows directly from the first.
 
  For the first claim, consider any $1\leq i < m$ and $v\in S_i$. Note that
  \begin{enumerate}[(1)]
    \item at time $t_i$ we create $L\cdot (i+1)\cdot a_{t_i}(v)$ descendants
      of $v$ at depth $L\cdot (i+1)$;
    \item $t_{i+1}-t_i = L\cdot (i+1)$.
  \end{enumerate}
  Because of this, no agent passing through the root can visit any descendant
  of $v$ at depth $L\cdot(i+1)$ before
  round $t_{i+1}$. On the other hand, the $a_{t_i}(v)$
  agents that could visit a descendant at this depth without passing through the root cannot visit
  \emph{all} descendants before round $t_{i+1}$. Thus at least one descendant at
  depth $L\cdot (i+1)$ must be unvisited at the end of round $t_{i+1}-1$.
\end{proof}

\begin{claim}\label{cl:si}
  For every $1\leq i < m$ we have the bounds
  \[ |S_i| \geq \frac{\alpha^i n}{2L} \geq \frac{1}{\alpha} \quad\text{and}\quad
  |S_i| \leq \frac{(2\alpha)^i n}{2L}. \]
\end{claim}

\begin{proof}
  By definition we have $\alpha = (2L/n)^{1/m} < 1$ and thus $\alpha^m = 2L/n$, which gives us
  \[ \frac{\alpha^in}{2L} \geq \frac{\alpha^{m-1}n}{2L} = 1/\alpha \]
  for all $1\leq i < m$.

  For the lower bound, note that  since $A_0$ contains only the root, we have
  $|K_1| = \ceil{n/(2L)}$. By the definition of $S_i$, we have
  $|S_i|\geq \alpha |K_i|$ for all $1\leq i< m$. Moreover, if $2\leq i< m$ then by Claim~\ref{cl:cl2} we have
  $|K_i| = |S_{i-1}|$. The lower bound then follows by induction.

  For the upper bound, note that $K_1 \leq n/(2L) + 1 \leq n/L$,
  where the last inequality uses $n\geq 2L$. Moreover, using $|K_i| \geq 1/\alpha$ we have $|S_i|
  \leq \alpha|K_i|+1 \leq 2\alpha |K_i|$ for all $1\leq i< m$. Finally, if $2\leq i < m$
  then $|K_i| = |S_{i-1}|$ by Claim~\ref{cl:cl2}, and the upper bound follows by induction.
\end{proof}

Since $|S_i|>0$ implies in particular that $A_{t_i-1}\neq V(T_{t_i-1})$, we conclude from Claim~\ref{cl:si} 
that the game does not stop before reaching round $t_{m-1} = L\cdot
\binom{m}{2} = t^*$. Moreover, from Claim~\ref{cl:height} and Claim~\ref{cl:si} we see
that $T_{t_{m-1}}$ is a tree with height $L\cdot m$. To complete the proof we
need to show that $|V(T_{t_{m-1}})| \leq
n$. We have
\begin{align}
  |V(T_{t_{m-1}})| & \leq \ceil{n/(2L)}\cdot L + 1 + \sum_{i=1}^{m-1}\sum_{v\in S_i}(L-1+L\cdot
  (i+1)\cdot a_{t_i}(v))\nonumber\\
  & \leq n/2 + L +1 + \sum_{i=1}^{m-1}L(i+1)\sum_{v\in S_i}a_{t_i}(v)
  + \sum_{i=1}^{m-1} |S_i| (L-1). \label{eq:sum}
\end{align}
To bound the double sum note that $|K_i|\geq |S_i|\geq 1/\alpha$ (Claim~\ref{cl:si}) implies
that $\ceil{\alpha |K_i|} \leq 2\alpha |K_i|$.
Note also that the sum  $\sum_{v \in K_i} a_{t_i}(v)$ in \eqref{eq:sum} is at
most $k$, as no two vertices  $u, v$ from $K_i$ are in the same subtree, i.e.,
$u \not \sim v$.
Since $S_i$ contains the  $\ceil{\alpha |K_i|} \leq 2\alpha |K_i|$
vertices of $K_i$ with least $a_{t_i}(v)$, we thus have
\[ \sum_{v\in S_i}a_{t_i}(v) \leq 2\alpha \sum_{v\in K_i}a_{t_i}(v) \leq 2\alpha k, \]
and therefore
\begin{align}
	\sum_{i=1}^{m-1}L(i+1)\sum_{v\in S_i}a_{t_i}(v) \leq L(m+1)^2\alpha k. \label{eq:doublesum}
\end{align}
To bound the simple sum in \eqref{eq:sum}, we use the upper bound from Claim~\ref{cl:si} and obtain
\begin{align} 
	\sum_{i=1}^{m-1} |S_i| (L-1) \leq (L-1) \sum_{i=1}^\infty \frac{(2\alpha)^in}{2L}
= \frac{L-1}{2L}\cdot 2 \alpha n \sum_{i=0}^\infty (2\alpha)^i
\leq \frac{2\alpha n}{2-4\alpha}. \label{eq:simplesum}
\end{align}
Combining \eqref{eq:sum} with \eqref{eq:doublesum} and \eqref{eq:simplesum}, we get
\begin{align} 
	|V(T_{t_{m-1}})| \leq n/2 + L +1 + L(m+1)^2\alpha k +\frac{2\alpha n}{2-4\alpha}. \label{eq:simplersum}
\end{align}
Since $n \geq L\cdot 16^m\geq 12L$ we have $L+1\leq 2L\leq n/6$.
By the definition $\alpha=(2L/n)^{1/m}$ and the assumption  
$k  \le n^{1 + 1/m} / (6L (m+1)^2 (2L)^{1/m})$ we have 
\[ L(m+1)^2\alpha k = L(m+1)^2(2L/n)^{1/m}k \leq n/6.\] 
Finally, $n\geq L\cdot 16^m$ implies
that $\alpha \leq 1/8$ and so
the last term in~\eqref{eq:simplersum} is also at most $n/6$. Hence $|V(T_{t_{m-1}})| \leq n/2 + 3n/6 =n$.
\end{proof}

\section{Consequences for competitiveness}

We now use Lemma~\ref{lemma:main} to derive consequences for best-possible
competitive ratios of collaborative tree exploration algorithms.
In the proofs below, $\log$ is always to the natural base $e$.

\thmLoglogCompRatio*

\begin{proof}
  By the result of Dynia et al.~\cite{DyniaLopuszanskiSchindelhauer/07} it suffices to consider the case where
  $k\geq \sqrt{n}$. Let $c>0$ be a constant and assume $k\leq n\log^c n$.
  We apply Lemma~\ref{lemma:main} with 
  $m= \ceil{\frac{\log n}{(8 + c) \log\log n}}$ and $L = \ceil{n/(mk)}$.
  Using $k\geq \sqrt{n}$, we have~$L = \mathcal{O}(\sqrt{n})$ and $m = o(\log n)$ and thus $n\geq L\cdot 16^m$ holds for
  sufficiently large $n$.
The lemma states that if
  \begin{equation}\label{eq:cr}
    k\leq \frac{n^{1+1/m}}{6L(m+1)^2(2L)^{1/m}}
  \end{equation}
  then there is a tree of height $D := Lm$ on which the strategy needs
  at least $L\binom{m}{2} = \Omega((n/k+D)\cdot \log k/\log\log k)$ rounds. To complete the proof,
we need to show that \eqref{eq:cr} holds for all $1\leq k \leq n\log^c n$.
We split the analysis to two cases. Let us first assume $k \ge n/ m$ and thus $L = 1$. This implies
\[
\frac{n^{1+1/m}}{6L(m+1)^2(2L)^{1/m}} \ge \frac{n^{1 + 1/m}}{24m^2} 
\ge \frac{n \log^{8  + c} n}{24 \log^2 n} \ge k,
\]
when $k \le n \log^c n$ and for sufficiently large~$n$.

Now we consider the case $k < n/ m$. 
Using that assumption and  the definition of $L$ we obtain $L (m+1)^2 \le 4m n/ k$ and $2L \le 4 n / (mk)$. Putting it all together we have
\[
\frac{n^{1+1/m}}{6L(m+1)^2(2L)^{1/m}}= \frac{n}{6L (m +1)^2} \left(\frac{n}{2L} \right)^{1/m}
\ge \frac{k}{24m} \left (\frac{mk}{4} \right) ^{1/m} \ge k,
\]
where the last inequality holds for $k\geq \sqrt{n}$ because, for sufficiently large~$n$,
\[(mk)^{1/m} \ge k^{1/m} \geq e^{\frac{\log n}{2m}} \geq
e^{\frac{(8+c)\log \log n}{4}} \geq (\log n)^2 \geq 100m.\qedhere \]

\end{proof}

\thmIterations*

\begin{proof}
  We choose $L=1$ and $m = \ceil{1/2\epsilon}$ in Lemma~\ref{lemma:main}.
  The claim is trivial unless $\epsilon<1/5$, so we can eliminate rounding and assume generously that $1/m\geq
  1.4\epsilon$.
  The condition $n\geq L\cdot 16^m$ is clearly satisfied for sufficiently large~$n$.

  By Lemma~\ref{lemma:main}, there is a tree $T$ of height $m$ that needs time
  $\binom{m}{2}\geq m/(5\epsilon)$ to be explored, provided the team has size
  at most (for $n$ sufficiently large)
  \[ k \leq n^{1+1.4\epsilon}/(12(m+1)^2) \leq m\cdot n^{1+\epsilon} = D\cdot n^{1+\epsilon}.\qedhere \]
\end{proof}

\thmSquaredRounds*

\begin{proof}
  We choose $L=1$ and $m = \ceil{\sqrt{\log n}\,}$ in Lemma~\ref{lemma:main}.
  Then $n\geq L\cdot 16^m$ holds for sufficiently large $n$. Note also that
  for sufficiently large $n$,
  \[\frac{n^{1+1/m}}{12(m+1)^2} = \Omega(n\cdot e^{\sqrt{\log n}}/\log n) \geq n. \]
  The lemma now states that there exists a tree $T$ of height $m$ such that the
  given strategy with $k=n$ agents needs at least $\binom{m}{2}$ rounds to explore
  $T$. Since for large enough $n$ we have $\binom{m}{2}\geq m^2/3$, this implies the theorem.
\end{proof}

\thmAnySublinearDiameter*

\begin{proof}
  Suppose that $D\leq n^{o(1)}$, i.e., $D = n^{1/f(n)}$, where $f(n)$ is a function
  which tends to infinity with $n$. Let $L=D$ and note that we have
  \begin{align*}
  \frac{16L^{1/m}}{n^{1/m}} \leq  \frac{L^{1+1/m} (m+1)^2}{n^{1/m}} \le 
  \frac{4 m^2 n^{2 / f(n)} }{n^{1/m}}.
  \end{align*}
  If we  choose  $m = m(n) =\omega(1)$ as a function growing sufficiently slowly such that
  we have $m  \le  \min \{ (f(n))^{1/2}, (\log n)^{1/2}\}$,  then the
  following is true: $$
  \frac{4 m^2 n^{2 / f(n)} }{n^{1/m}} = 4 \cdot  e^{2 \log m  + 2 (\log n) / f(n) - \log n / m} \to 0.
  $$
 This implies $16 L^{1/m} = o(n^{1/m})$ and $L^{1+1/m} (m+1)^2 = o(n^{1/m})$.
  In particular, $n\geq L\cdot 16^m$ for sufficiently large $n$. Moreover,
  if $n$ is large enough then $k=\Theta(n)$ implies
  \[ \frac{n^{1+1/m}}{6L(m+1)^2(2L)^{1/m}} = \frac{n^{1+1/m}}{o(n^{1/m})} \geq k. \]
  By Lemma~\ref{lemma:main}, there
  exists a tree $T$ with height $Lm \geq D$ on which the strategy needs $L\binom{m}{2} = \omega(Lm)$
  rounds. Since $k= \Theta(n)$, the offline optimum is $\mathcal O(Lm+n/k) = \mathcal O(Lm)$, so the competitive ratio
  on the set of trees of height at least $D$ is $\omega(1)$, as claimed.
\end{proof}

\section{Conclusions}

In this paper we presented new lower bounds for collaborative tree exploration.
Including our results, the following bounds are now known.
For~$k = \mathcal{O}(1)$ or $k \geq D\cdot n^{1+\varepsilon}$ agents, a competitive ratio of~$\Theta(1)$ can be achieved.
For~$\omega(1) \leq k \leq n \log^c n$, the best-possible competitive ratio is
bounded by~$\Omega(\log k / \log \log k)$, and no constant competitive ratio is possible when $n \log^{c} n \leq k
\leq D\cdot n^{1+o(1)}$.
On the other hand, the best exploration algorithms for trees in the domain~$k
\leq D\cdot n^{1+o(1)}$ stay close to the trivial competitive ratio of~$k$ (the
best ratios are~$k / \log k$ and~$k^{o(1)}$, depending on the domain).

In summary, we now fully understand the domain where constant competitive ratios are possible, but,
outside this domain, a wide gap persists.

\section*{Acknowledgments}

We would like to thank Rajko Nenadov for useful discussions.

\pagebreak

\bibliographystyle{abbrv}
\bibliography{explorationbib}

\end{document}